\documentclass[10pt]{article}

\usepackage{hyperref}
\usepackage{enumerate}
\usepackage{amsmath,amsthm,amssymb}
\usepackage{fullpage}
\usepackage[all=normal,bibliography=tight]{savetrees}
\def\boxit#1{\vbox{\hrule\hbox{\vrule\kern4pt
  \vbox{\kern1pt#1\kern1pt}
\kern2pt\vrule}\hrule}}

\newtheorem{lemma}{Lemma}
\newtheorem{theorem}[lemma]{Theorem}
\newtheorem{definition}[lemma]{Definition}

\newtheorem{claim}[lemma]{Claim}
\newcommand{\Oh}{\mathcal{O}}

\usepackage{graphicx}
\usepackage[absolute]{textpos}

\def\cqedsymbol{\ifmmode$\lrcorner$\else{\unskip\nobreak\hfil
\penalty50\hskip1em\null\nobreak\hfil$\lrcorner$
\parfillskip=0pt\finalhyphendemerits=0\endgraf}\fi}

\newcommand{\cqed}{\renewcommand{\qed}{\cqedsymbol}}

\begin{document}

\title{Multi-budgeted directed cuts%
\thanks{This research is a part of projects that have received funding from the European Research Council (ERC) under the European Union's Horizon 2020 research and innovation programme under grant agreements No 714704 (S. Li and M. Pilipczuk), 280152 (D. Marx), and 725978 (D. Marx).
M. Wahlstr\"{o}m is supported by EPSRC grant EP/P007228/1.}}

\author{Stefan Kratsch\thanks{Institut f\"{u}r Informatik, Humboldt-Universit\"{a}t zu Berlin, Germany, \texttt{kratsch@informatik.hu-berlin.de}.} \and
  Shaohua Li\thanks{Institute of Informatics, University of Warsaw, Poland, \texttt{Shaohua.Li@mimuw.edu.pl}.}
  \and
  D\'{a}niel Marx\thanks{Institute for Computer Science and Control, Hungarian Academy of Sciences (MTA SZTAKI), Hungary, \texttt{dmarx@cs.bme.hu}.}
  \and
  Marcin Pilipczuk\thanks{Institute of Informatics, University of Warsaw, Poland, \texttt{malcin@mimuw.edu.pl}.}
  \and
  Magnus Wahlstr\"{o}m\thanks{Royal Holloway, University of London, UK, \texttt{Magnus.Wahlstrom@rhul.ac.uk}.}}

  \date{}

\maketitle

\begin{textblock}{20}(0, 12.0)
\includegraphics[width=40px]{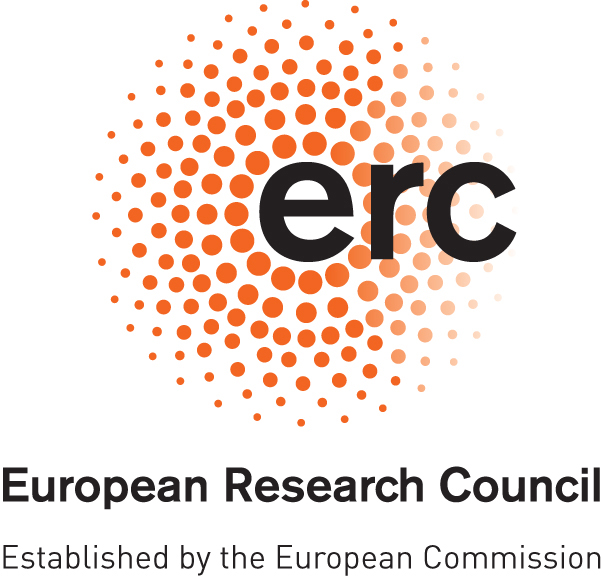}%
\end{textblock}
\begin{textblock}{20}(0, 12.9)
\includegraphics[width=40px]{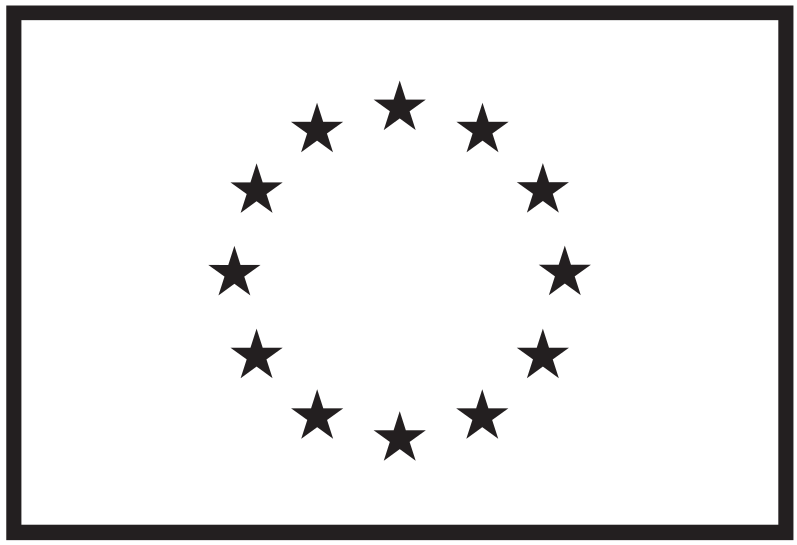}%
\end{textblock}

\begin{abstract}
In this paper, we study multi-budgeted variants of the classic minimum cut problem and graph separation problems that turned out to be important in parameterized complexity:
\textsc{Skew Multicut} and \textsc{Directed Feedback Arc Set}.
In our generalization, we assign colors $1,2,...,\ell$ to some edges and give separate budgets $k_{1},k_{2},...,k_{\ell}$ for colors $1,2,...,\ell$.
For every color $i\in \{1,...,\ell\}$, let $E_{i}$ be the set of edges of color $i$.
The solution $C$ for the multi-budgeted variant of a graph separation problem not only needs to satisfy the usual separation requirements
(i.e., be a cut, a skew multicut, or a directed feedback arc set, respectively), but also needs to satisfy that $|C\cap E_{i}|\leq k_{i}$ for every $i\in \{1,...,\ell\}$.

Contrary to the classic minimum cut problem, the multi-budgeted variant turns out to be NP-hard even for $\ell = 2$.
We propose FPT algorithms parameterized by $k=k_{1}+...+k_{\ell}$ for all three problems.
To this end, we develop a branching procedure for the multi-budgeted minimum cut problem that measures the progress of the algorithm not by reducing $k$ as usual, by but elevating the capacity
of some edges and thus increasing the size of maximum source-to-sink flow.
Using the fact that a similar strategy is used to enumerate all important separators of a given size, we merge this process with the flow-guided branching and show an FPT bound on the number of
(appropriately defined) important multi-budgeted separators. This allows us to extend our algorithm to the \textsc{Skew Multicut} and \textsc{Directed Feedback Arc Set} problems.

Furthermore, we show connections of the multi-budgeted variants with weighted variants of the directed cut problems and the \textsc{Chain $\ell$-SAT} problem, whose parameterized complexity
remains an open problem. We show that these problems admit a bounded-in-parameter number of ``maximally pushed'' solutions (in a similar spirit as important separators are maximally pushed),
giving somewhat weak evidence towards their tractability.
\end{abstract}

\section{Introduction}

Graph separation problems are important topics in both theoretical area and applications. Although the famous minimum cut problem is known to be polynomial-time solvable, many well-known variants are NP-hard, which are intensively studied from the point of view of approximation \cite{DBLP:conf/stoc/AgarwalAC07,DBLP:journals/siamdm/ChekuriGN06,DBLP:journals/algorithmica/EvenNSS98,DBLP:journals/siamcomp/GargVY96,DBLP:journals/jal/GargVY04,
DBLP:journals/mor/KargerKSTY04} and, what is more relevant for this work, parameterized complexity.

The notion of important separators, introduced by Marx~\cite{DBLP:journals/tcs/Marx06},
turned out to be fundamental for a number
of graph separation problems such as \textsc{Multiway Cut}~\cite{DBLP:journals/tcs/Marx06}, \textsc{Directed Feedback Vertex Set}~\cite{DBLP:journals/jacm/ChenLLOR08}, or \textsc{Almost $2$-CNF SAT}~\cite{DBLP:journals/jcss/RazgonO09}.
Further work, concerning mostly undirected graphs, resulted in a wide range of involved algorithmic
techniques: applications of matroid techniques~\cite{ms2,ms1},
shadow removal~\cite{dir-mwc,DBLP:journals/siamcomp/MarxR14}, randomized
contractions~\cite{rand-contr}, LP-guided branching~\cite{mwc-a-lp,sylvain,DBLP:conf/icalp/Iwata17,magnus},
and treewidth reduction~\cite{tw-red}, among others.

From the above techniques, only the notion of important separators and the related
technique of shadow removal generalizes to directed graphs,
          giving FPT algorithms for \textsc{Directed Feedback Arc Set}~\cite{DBLP:journals/jacm/ChenLLOR08},
          \textsc{Directed Multiway Cut}~\cite{dir-mwc},
          and~\textsc{Directed Subset Feedback Vertex Set}~\cite{dsfvs}.
As a result, the parameterized complexity of a number of important graph separation
problems in directed graphs remains open, and the quest to investigate them has been
put on by the third author in a survey from 2012~\cite{marx-survey}.
Since the publication of this survey, two negative answers have been obtained.
Two authors of this work showed that \textsc{Directed Multicut} is W[1]-hard even
for four terminal pairs (leaving the case of three terminal pairs open)~\cite{DBLP:conf/soda/PilipczukW16},
while Lokshtanov et al.~\cite{LokshtanovRSZ17} showed intractability of
\textsc{Directed Odd Cycle Transversal}.

During an open problem session at Recent Advancements in Parameterized Complexity school
(December 2017)~\cite{RAPC}, Saurabh posed the question of parameterized complexity of
a weighted variant of \textsc{Directed Feedback Arc Set}, where
given a directed edge-weighted graph $G$, an integer $k$, and a target weight $w$, the goal is to find
a set $X \subseteq E(G)$ such that $G-X$ is acyclic and $X$ is of cardinality at most $k$ and weight at most $w$.
Consider a similar problem \textsc{Weighted $st$-cut}:
given a directed graph $G$ with positive edge weights and two distinguished vertices $s,t \in V(G)$, an integer $k$, and a target weight $w$,
decide if $G$ admits an $st$-cut of cardinality at most $k$ and weight at most $w$.
The parameterized complexity of this problem parameterized by $k$ is open even if $G$ is restricted to be acyclic,
    while with this restriction the problem can easily be reduced to \textsc{Directed Feedback Arc Set}
    (add an arc $(t,s)$ of prohibitively large weight).

The \textsc{Weighted $st$-cut} problem becomes similar to another directed graph cut problem, identified in~\cite{DBLP:journals/algorithmica/ChitnisEM17},
namely \textsc{Chain $\ell$-SAT}.
While this problem is originally formulated in CSP language, the graph formulation is as follows:
given a directed graph $G$ with a partition of edge set
$E(G) = P_1 \uplus P_2 \uplus \ldots \uplus P_m$ such that each $P_i$ is an edge set
of a simple path of length at most $\ell$ (the input paths could have common nodes), an integer $k$, and two vertices $s,t \in V(G)$,
find an $st$-cut $C \subseteq E(G)$ such that $|\{i | C \cap P_i \neq \emptyset\}| \leq k$.
This problem can easily be seen to be equivalent to minimum $st$-cut problem (and
    thus polynomial-time solvable) for $\ell \leq 2$, but is NP-hard for $\ell \geq 3$
and its parameterized complexity (with $k$ as a parameter) remains an open problem.

In this paper we make progress towards resolving the question of parameterized complexity
of the two aforementioned problems: weighted $st$-cut problem (in general digraphs, not necessary acyclic ones)
and \textsc{Chain $\ell$-SAT}.
Our contribution is twofold.

\paragraph*{Multi-budgeted variant}
We define a \emph{multi-budgeted} variant of a number of cut problems
(including the minimum cut problem) and show its fixed-parameter tractability.
In this variant, the edges of the graph are colored with $\ell$ colors, and the input
specifies separate budgets for each color.
More formally, we primarily consider the following problem.

\begin{quote}
\textsc{Multi-budgeted cut}\\
\textbf{Input:} A directed graph $G$, two disjoint sets of vertices $X,Y\subseteq V(G)$, an integer $\ell$, and for every $i \in \{1,2,\ldots,\ell\}$ a set $E_i \subseteq E(G)$ and an integer
$k_i \geq 1$.\\
\textbf{Question:} Is there a set of arcs $C\subseteq \bigcup_{i=1}^{\ell}E_{i}$ such that there is no directed $X-Y$ path in $G\setminus C$ and for every $i\in [\ell]$, $|C\cap E_i|\leq k_{i}$.\\
\end{quote}

Similarly we define multi-budgeted variants of \textsc{Directed Feedback Arc Set} and
\textsc{Skew Multicut} (see Section~\ref{sec:imp} for formal definitions).

We observe that \textsc{Multi-budgeted cut} for $\ell=2$ reduces to \textsc{Weighted $st$-cut} as follows.
Let $(G,X,Y,E_1,E_2,k_1,k_2)$ be a \textsc{Multi-budgeted cut} instance for $\ell=2$.
First, observe that we may assume that $E_1 \cap E_2 = \emptyset$, as we can replace every edge $e \in E_1 \cap E_2$ with two copies $e_1 \in E_1 \setminus E_2$
and $e_2 \in E_2 \setminus E_1$.
Second, construct
an equivalent \textsc{Weighted $st$-cut} instance $(G',s,t,k,w)$ as follows. To construct $G'$, first add two vertices $s,t$ to $G$
and edges $\{(s,x) | x \in X\}$ and $\{(y, t) | y \in Y\}$ of prohibitively large weight. Assign also prohibitively large weight
to every edge $e \in E(G) \setminus (E_1 \cup E_2)$. Assign weight $(k_1+1)k_2+1$ to every edge $e \in E_1$.
For every edge $e \in E_2$, add $k_1+1$ copies of $e$ to $G'$ of weight $1$ each.
Finally, set $k := (k_1+1) \cdot k_2 + k_1$ as the cardinality bound and $w := k_1((k_1+1)k_2+1) + (k_1+1)k_2$ as the target weight.
The equivalence of the instances follows from the fact that the cardinality bound allows to pick in the solution at most $k_2$ bundles of $k_1+1$ copies
of an edge of $E_2$, while the weight bound allows to pick only $k_1$ edges of $E_1$.

Thus, \textsc{Multi-budgeted cut} for $\ell=2$ corresponds to the case of \textsc{Weighted $st$-cut}
where the weights are integral and both target cardinality and weight are bounded in parameter.%
\footnote{For a reduction in the other direction, replace every arc $e$ of weight $\omega(e)$ with one copy of color $1$ and $\omega(e)$ copies of color $2$, and set
  budgets $k_1 = k$ and $k_2 = w$.}
  This connection was our primary motivation to study the multi-budgeted variants of the cut problems.

Contrary to the classic minimum cut problem, in Section~\ref{sec:np} we note that \textsc{Multi-budgeted Cut} becomes
NP-hard for $\ell \geq 2$.%
\footnote{We believe this problem must have been formulated already before and proven to be NP-hard.
  However, we were not able to find it in the literature, and thus we provide our own simple NP-hardness reduction for completeness.}
We show that \textsc{Multi-budgeted Cut} is FPT when parameterized by $k=k_{1}+...+k_{\ell}$.
For this problem, our branching strategy is as follows.
First, note that in the problem definition we assume that each $k_i$ is positive, and thus $\ell \leq k$.
A standard application of the Ford-Fulkerson algorithm gives a minimum $XY$-cut $C$ of size
$\lambda$ and $\lambda$ edge-disjoint $X-Y$ paths $P_1,P_2,\ldots,P_\lambda$.
If $C$ is a solution, then we are done.
Similarly, if $\lambda > k$, then there is no solution.
Otherwise, we branch which colors of the sought solution
should appear on each paths $P_j$; that is, for every $i \in [\ell]$ and $j \in [\lambda]$,
we guess if $P_j \cap E_i$ contains an edge of the sought solution, and in each guess
assign infinite capacities to the edges of wrong color.
If this change increased the size of a maximum flow from $X$ to $Y$, then we can charge the
branching step to this increase, as the size of the flow cannot exceed $k$.
The critical insight is that if the size of the maximum flow does not increase
(i.e., $P_1,\ldots,P_\lambda$ remains a maximum flow), then a corresponding minimum cut is necessarily
a solution. As a result, we obtain the following.

\begin{theorem}\label{thm:cut}
\textsc{Multi-budgeted Cut} admits an FPT algorithm with running time bound
$\Oh(2^{k^2 \ell} \cdot k \cdot (|V(G)|+|E(G)|))$ where $k = \sum_{i=1}^\ell k_i$.
\end{theorem}

The charging of the branching step to a flow increase appears also in the classic
argument for bound of the number of important separators~\cite{DBLP:journals/jacm/ChenLLOR08} (see also~\cite[Chapter 8]{DBLP:books/sp/CyganFKLMPPS15}).
We observe that our branching algorithm can be merged with this procedure, yielding
a bound (as a function of $k$) and enumeration procedure
of naturally defined multi-budgeted important separators.
This in turn allows us to generalize our FPT algorithm
to \textsc{Multi-budgeted Skew Multicut} and \textsc{Multi-budgeted Directed Feedback Arc Set}.

\begin{theorem}\label{thm:dfas}
\textsc{Multi-budgeted Skew Multicut} and \textsc{Multi-budgeted Directed Feedback Arc Set} admit FPT algorithms
with running time bound $2^{\Oh(k^3 \log k)} (|V(G)|+|E(G)|)$ where $k = \sum_{i=1}^\ell k_i$.
\end{theorem}

The arguments for \textsc{Multi-budgeted Cut} are presented in Section~\ref{sec:cut}.
The generalization for important separators is contained in Section~\ref{sec:imp}.

\paragraph*{Bound on the number of pushed solutions}
While we are not able to establish fixed-parameter tractability of
the weighted variant of the minimum cut problem (even in acyclic graphs) nor of \textsc{Chain $\ell$-SAT},
we show the following graph-theoretic statement.
Consider a directed graph $G$ with two distinguished vertices $s,t \in V(G)$.
For two (inclusion-wise) minimal $st$-cuts $C_1$, $C_2$ we say
that $C_1$ \emph{is closer to $t$} than $C_2$ if every vertex
reachable from $s$ in $G-C_2$ is also reachable from $s$ in $G-C_1$.
A classic submodularity argument implies that there is exactly one closest
to $t$ minimum $st$-cut, while the essence of the notion of important separators
is the observation that there is bounded-in-$k$ number of minimal separators of
cardinality at most $k$ that are closest to $t$.
In Section~\ref{sec:bound} we show a similar existential statement for the two discussed problems.

\begin{theorem}\label{thm:bnd-weights}
For every integer $k$ there exists an integer $g$ such that the following holds.
Let $G$ be a directed graph with positive edge weights and two distinguished vertices $s,t \in V(G)$.
Let $\mathcal{F}$ be a family of all $st$-cuts that are of minimum weight
among all (inclusion-wise) minimal $st$-cuts of cardinality at most $k$.
Let $\mathcal{G} \subseteq \mathcal{F}$ be the family of those cuts $C$ such that no other
cut of $\mathcal{F}$ is closer to $t$.
Then $|\mathcal{G}| \leq g$.
\end{theorem}
\begin{theorem}\label{thm:bnd-chain}
For every integers $k,\ell$ there exists an integer $g'$ such that the following holds.
Let $I := (G,s,t,(P_i)_{i=1}^m,k)$ be a \textsc{Chain $\ell$-SAT} instance that is a yes-instance
but $(G,s,t,(P_i)_{i=1}^m,k-1)$ is a no-instance.
Let $\mathcal{F}$ be a family of all (inclusion-wise) minimal solutions to $I$
and let $\mathcal{G} \subseteq \mathcal{F}$ be the family of those cuts $C$ such that no other
cut of $\mathcal{F}$ is closer to $t$.
Then $|\mathcal{G}| \leq g'$.
\end{theorem}
Unfortunately, our proof is purely existential, and does not yield an enumeration procedure of the ``closest to $t$'' solutions.


\section{Preliminaries}
For an integer $n$, we denote $[n] = \{1,2,\ldots,n\}$.
For a directed graph $G$, we use $V(G)$ to represent the set of vertices of $G$ and $E(G)$ to represent the set of directed edges of $G$.
In all multi-budgeted problems, the directed graph $G$ comes with sets $E_i \subseteq E(G)$ for $i\in[\ell]$ which we refer as \emph{colors}.
That is, an edge $e$ is \emph{of color $i$} if $e \in E_i$, and of \emph{no color} if $e \in E(G) \setminus \bigcup_{i=1}^\ell E_i$.
Note that an edge may have many colors, as we do not insist on the sets $E_i$ being pairwise disjoint.

Let $X$ and $Y$ be two disjoint vertex sets in a directed graph $G$, an $XY$-\emph{cut} of $G$ is a set of edges $C$ such that every directed path from a vertex in $X$ to a vertex in $Y$ contains an edge of $C$.
A cut $C$ is \emph{minimal} if no proper subset of $C$ is an $XY$-cut, and \emph{minimum} if $C$ is of minimum possible cardinality.
Let $C$ be an $XY$-cut and let $R$ be the set of vertices reachable from $X$ in $G\setminus C$.
We define $\delta^{+}(R)=\{(u,v)\in E(G)|u\in R$ and $v\notin R\}$ and note that if $C$ is minimal, then $\delta^+(R) = C$.

Let $(G,X,Y,\ell,(E_i,k_i)_{i=1}^\ell)$ be a \textsc{Multi-budgeted cut} instance and let $C$ be an $XY$-cut.
We say that $C$ is \emph{budget-respecting} if $C \subseteq \bigcup_{i=1}^\ell E_i$ and $|C \cap E_i| \leq k_i$ for every $i \in [\ell]$.
For a set $Z \subseteq E(G)$ we say that $C$ is \emph{$Z$-respecting} if $C \subseteq Z$.
In such contexts, we often call $Z$ the \emph{set of deletable edges}.
An $XY$-cut $C$ is a \emph{minimum $Z$-respecting cut} if it is a $Z$-respecting $XY$-cut of minimum possible cardinality
among all $Z$-respecting $XY$-cuts.

Our FPT algorithms start with $Z = \bigcup_{i=1}^\ell E_i$ and in branching steps shrink the set $Z$ to reduce the search space.
We encapsulate our use of the classic Ford-Fulkerson algorithm in the following statement.
\begin{theorem}\label{thm:ff}
Given a directed graph $G$, two disjoint sets $X,Y \subseteq V(G)$, a set $Z \subseteq E(G)$, and an integer $k$,
      one can in $\Oh(k(|V(G)|+|E(G)|))$ time either find the following objects:
      \begin{itemize}
      \item $\lambda$ paths $P_1,P_2,\ldots,P_\lambda$ such that every $P_i$ starts in $X$ and ends in $Y$,
     and every edge $e \in Z$ appears on at most one path $P_i$;
      \item a set $B \subseteq Z$ consisting of all edges of $G$ that participate in some minimum $Z$-respecting $XY$-cut;
      \item a minimum $Z$-respecting $XY$-cut $C$ of size $\lambda$ that is closest to $Y$ among all minimum $Z$-respecting $XY$-cuts;
      \end{itemize}
      or correctly conclude that there is no $Z$-respecting $XY$-cut of cardinality at most $k$.
\end{theorem}
\begin{proof}
Assign capacity $1$ to every edge of $Z$ and capacity $+\infty$ to every edge not in $Z$. Run $k+1$ rounds
of the Ford-Fulkerson algorithm. If the final flow exceeded $k$, return that there is no $Z$-respecting $XY$-cut of cardinality at most $k$.
Otherwise, decompose the final flow into unit flow paths $P_1,\ldots,P_\lambda$ in a standard manner.
For the set $B$, observe that $B$ consists of exactly those edges that are fully saturated in the flow network, and their
reverse counterparts in the residual network are not contained in a single strongly connected component of the residual network (and thus
    can be discovered in linear time).
Finally, observe that the sought cut $C$ consists of the last edge of $B$ on every path $P_i$.
\end{proof}

\section{Multi-budgeted cut}\label{sec:cut}

We now give an FPT algorithm parameterized by $k=\Sigma_{i=1}^{\ell}k_{i}$ for the \textsc{Multi-budgeted cut} problem.
We follow a branching strategy that recursively reduces a set $Z$ of deletable edges.
That is, we start with $Z = \bigcup_{i=1}^\ell E_i$ (so that every solution is initially $Z$-respecting) and in each recursive
step, we look for a $Z$-respecting solution and reduce the set $Z$ in a branching step.

Consider a recursive call where we look for a $Z$-respecting solution to the input \textsc{Multi-budgeted cut}
instance $(G,X,Y,\ell,(E_i,k_i)_{i=1}^\ell)$. That is, we look for a $Z$-respecting budget-respecting cut.
We apply Theorem~\ref{thm:ff} to it. If it returns that there is no $Z$-respecting $XY$-cut of size at most $k$, we terminate the current branch, as there is no solution.
Otherwise, we obtain the paths $P_1,P_2,\ldots,P_\lambda$, the set $B$ (which we will not use in this section), and the cut $C$.

If $C$ is budget-respecting, then it is a solution and we can return it.
Otherwise, we perform the following branching step.
We iterate over all tuples $(A_{1},...,A_{\ell})$ such that for every $i\in [\ell]$, $A_{i}\subseteq [\lambda]$ and $|A_{i}|\leq k_i$.
$A_{i}$ represents the subset of paths $P_{1},...,P_{\lambda}$ on which at least one edge of color $i$ is in the solution for each $i\in [\ell]$.
For those edges of color $i$ which are on the paths not indicated by $A_i$, they are not in the solution.
Thus we can safely delete them from $Z$.
More formally, for every $i \in [\ell]$ and $j \in [\lambda] \setminus A_i$, we remove from $Z$ all edges of $E(P_j) \cap E_i$.
We recurse on the reduced set $Z$.
A pseudocode is available in Figure~\ref{fig1}.

\begin{figure}[htbp]
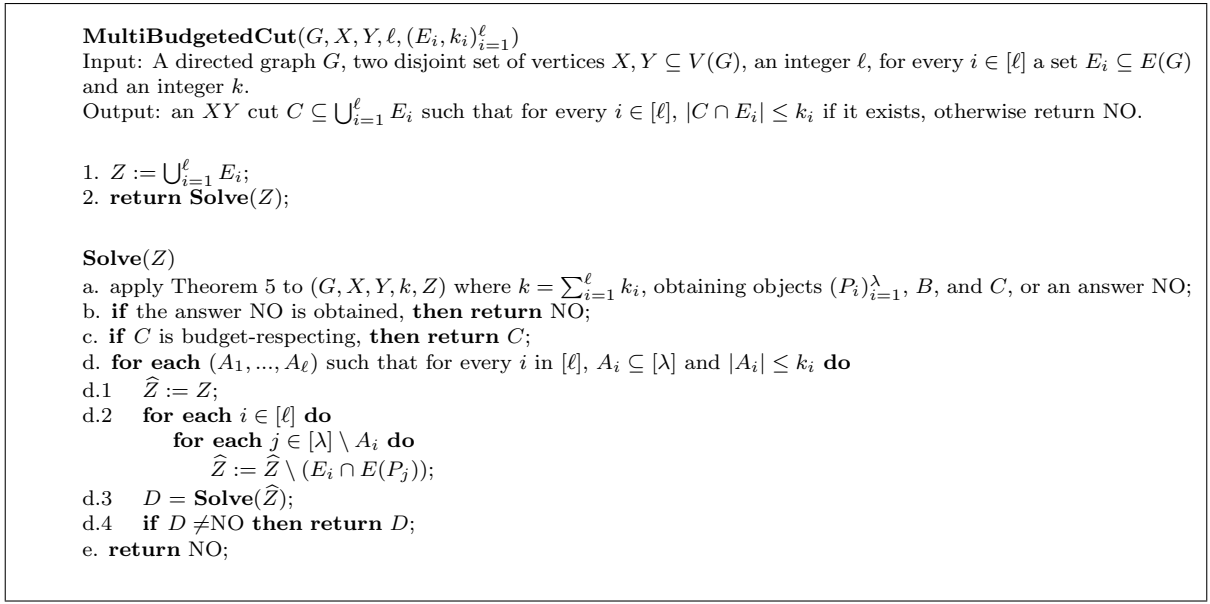

\setbox4=\vbox{\hsize32pc \strut \begin{quote}
\vspace*{-5mm} \footnotesize
\textbf{MultiBudgetedCut}($G,X,Y,\ell,(E_i,k_i)_{i=1}^\ell$) \\
Input: A directed graph $G$, two disjoint set of vertices $X,Y\subseteq V(G)$, an integer $\ell$, for every $i \in [\ell]$ a set $E_i \subseteq E(G)$
and an integer $k$.\\
Output: an $XY$ cut $C\subseteq \bigcup_{i=1}^{\ell}E_{i}$ such that for every $i\in [\ell]$, $|C\cap E_i|\leq k_{i}$ if it exists, otherwise return NO.\\

1. $Z := \bigcup_{i=1}^\ell E_i$; \\
2. \textbf{return} \textbf{\textbf{Solve}}($Z$); \\

\textbf{Solve}($Z$) \\
a. apply Theorem~\ref{thm:ff} to $(G,X,Y,k,Z)$ where $k = \sum_{i=1}^\ell k_i$, obtaining objects
$(P_i)_{i=1}^\lambda$, $B$, and $C$, or an answer NO; \\
b. \textbf{if} the answer NO is obtained, \textbf{then} \textbf{return} NO;  \\
c. \textbf{if} $C$ is budget-respecting, \textbf{then} \textbf{return} $C$;    \\
d. \textbf{for each} $(A_{1},...,A_{\ell})$ such that for every $i$ in $[\ell]$, $A_{i}\subseteq [\lambda]$ and $|A_{i}|\leq k_i$ \textbf{do} \\
d.1 \hspace*{2mm} $\widehat{Z}:=Z$;   \\
d.2 \hspace*{2mm} \textbf{for each} $i\in [\ell]$ \textbf{do}  \\
\hspace*{11mm} \textbf{for each} $j\in [\lambda]\setminus{A_{i}}$ \textbf{do} \\
\hspace*{16mm} $\widehat{Z} := \widehat{Z} \setminus (E_i \cap E(P_j))$; \\
d.3 \hspace*{2mm} $D$ = \textbf{Solve}($\widehat{Z}$); \\
d.4 \hspace*{2mm} \textbf{if} $D\neq $NO \textbf{then} \textbf{return} $D$;\\
e. \textbf{return} NO;  \\
\end{quote}

\vspace*{-6mm}
\strut}
$$\boxit{\box4}$$
\vspace*{-9mm}
\caption{FPT algorithm for \textsc{Multi-budgeted cut}} \label{fig1}
\end{figure}

\begin{theorem}
The algorithm in Figure~\ref{fig1} for \textsc{Multi-budgeted cut} is correct and runs in time $O(2^{\ell k^{2}}\cdot k\cdot (|V(G)|+|E(G)|))$ where $k=\Sigma_{i=1}^{\ell}k_{i}$.
\end{theorem}
\begin{proof}
We prove the correctness of the algorithm by showing that it returns a solution if and only if the input instance is a yes-instance. The "only if" direction is obvious, as the algorithm returns only $Z$-respecting budget-respecting $XY$-cuts and $Z \subseteq \bigcup_{i=1}^\ell E_i$ in each recursive call.

We prove the correctness for the "if" direction. Let $C_0$ be a solution, that is, a budget-respecting $XY$-cut.
In the initial call to \textbf{Solve}, $C_0$ is $Z$-respecting. It suffices to inductively show that in each call to \textbf{Solve} such that $C_0$
is $Z$-respecting, either the call returns a solution, or $C_0$ is $\widehat{Z}$-respecting for at least one of the subcalls.
Since $C_0$ is $Z$-respecting, the application of Theorem~\ref{thm:ff} returns objects $(P_i)_{i=1}^\lambda$, $B$, and $C$.
If $C$ is budget-respecting, then the algorithm returns it and we are done.
Otherwise, consider the branch $(A_1,A_2,\ldots,A_\ell)$ where $A_i = \{j | E(P_j) \cap E_i \cap C_0 \neq \emptyset\}$.
Since $C_0$ is budget-respecting, $C_0 \subseteq Z$, and no edge of $Z$ appears on more than one path $P_j$, we have $|A_i| \leq k_i$ for every $i \in [\ell]$.
Thus, $(A_1,A_2,\ldots,A_\ell)$ is a branch considered by the algorithm.
In this branch, the algorithm refines the set $Z$ to $\widehat{Z}$.
By the definition of $A_i$, for every $i \in [\ell]$ and $j \in [\lambda] \setminus A_i$, we have $C_0 \cap E_i \cap E(P_j) = \emptyset$.
Consequently, $C_0$ is $\widehat{Z}$-respecting and we are done.

For the time bound, the following observation is crucial.
\begin{claim}\label{cl:cut}
Consider one recursive call \textbf{Solve}$(Z)$ where the application of Theorem~\ref{thm:ff} in line~(a) returned objects $(P_i)_{i=1}^\lambda$, $B$ and $C$.
Assume that in some recursive subcall \textbf{Solve}$(\widehat{Z})$ invoked in line~(d.3) (Figure~\ref{fig1}),
the subsequent application of Theorem~\ref{thm:ff} in line~(a) of the subcall returned a cut of the same size, that is,
the algorithm of Theorem~\ref{thm:ff} returned a cut $\widehat{C}$ of size $\widehat{\lambda} = \lambda$.
Then the cut $\widehat{C}$ is budget-respecting and, consequently, is returned in line~(c) of the subcall.
\end{claim}
\begin{proof}
Since $|\widehat{C}| = \lambda$ is a $\widehat{Z}$-respecting $XY$-cut, $\widehat{Z} \subseteq Z$,
and every edge $e \in Z$ appears on at most one path $P_i$, we have that $\widehat{C}$ consists of exactly one edge of $\widehat{Z}$ on every path $P_i$,
that is, $\widehat{C} = \{e_1,e_2,\ldots,e_\lambda\}$ and $e_j \in E(P_j) \cap \widehat{Z}$ for every $j \in [\lambda]$.
In other words, the paths $(P_j)_{j=1}^\lambda$ still correspond to a maximum flow from $X$ to $Y$ with edges of $\widehat{Z}$ being of unit capacity
 and edges outside $\widehat{Z}$ of infinite capacity because $(P_j)_{j=1}^\lambda$ are paths satisfying that any two of them are disjoint on $\widehat{Z}\subseteq Z$ and 
 $\lambda$ is still equal to the size of the maximum flow.
If $e_j \in E_i$ for some $j \in [\lambda]$ and $i \in [\ell]$, then by the construction of set $\widehat{Z}$, we have $j \in A_i$.
Consequently, $|\{ j | e_j \in E_i\}| \leq |A_i| \leq k_i$ for every $i \in [\ell]$, and thus $\widehat{C}$ is budget-respecting.
\cqed\end{proof}
Claim~\ref{cl:cut} implies that the depth of the search tree is bounded by $k$, as the algorithm terminates when $\lambda$ exceeds $k$.
At every step, there are at most $(2^{\lambda})^{\ell}\leq (2^{k})^{\ell}$ different tuples $(A_{1},...,A_{\ell})$ to consider.
Consequently, there are $\Oh(2^{(k-1)k\ell})$ nodes of the search tree that enter the loop in line~(d) and $\Oh(2^{k^2\ell})$ nodes
that invoke the algorithm of Theorem~\ref{thm:ff}.
As a result, the running time of the algorithm is $\Oh(2^{\ell k^{2}}\cdot k\cdot (|V(G)|+|E(G)|))$.
\end{proof}

\section{Multi-budgeted important separators with applications}\label{sec:imp}

Similar to the concept of important separators proposed by Marx \cite{DBLP:journals/tcs/Marx06} (see also \cite[Chapter 8]{DBLP:books/sp/CyganFKLMPPS15}),
we define \emph{multi-budgeted important separators} as follows.
\begin{definition}
Let $(G,X,Y,\ell,(E_i,k_i)_{i=1}^\ell)$ be a \textsc{Multi-budgeted cut} instance and let $Z \subseteq \bigcup_{i=1}^\ell E_i$ be a set of deletable edges.
Let $C_{1},C_{2}$ be two minimal $Z$-respecting budget-respecting $XY$-cuts. We say that $C_{1}$ \emph{dominates} $C_{2}$ if
\begin{enumerate}
\item every vertex reachable from $X$ in $G-C_{2}$ is also reachable from $X$ in $G-C_{1}$;
\item for every $i\in [\ell]$, $|C_{1}\cap E_{i}|\leq |C_{2}\cap E_{i}|$.
\end{enumerate}
We say that $\widehat{C}$ is an \emph{important $Z$-respecting budget-respecting $XY$-cut}
if $\widehat{C}$ is a minimal $Z$-respecting budget-respecting $XY$-cut and no other minimal $Z$-respecting budget-respecting $XY$-cut dominates $\widehat{C}$.
$\widehat{C}$ is an \emph{important budget-respecting $XY$-cut} if it is an \emph{important $Z$-respecting budget-respecting $XY$-cut} for
$Z = \bigcup_{i=1}^\ell E_i$.
\end{definition}

Chen et al.~\cite{DBLP:journals/jacm/ChenLLOR08} showed an enumeration procedure for (classic) important separators using similar charging scheme
as the one of the previous section. Our main result in this section is a merge of the arguments from the previous section
with the arguments of Chen et al., yielding the following theorem.
\begin{theorem}\label{thm:imp-enum}
Let $(G,X,Y,\ell,(E_i,k_i)_{i=1}^\ell)$ be a \textsc{Multi-budgeted cut} instance, let $Z \subseteq \bigcup_{i=1}^\ell E_i$ be a set of deletable edges,
and denote $k = \sum_{i=1}^\ell k_i$.
Then one can in $2^{\Oh(k^2 \log k)} (|V(G)|+|E(G)|)$ time enumerate
    a family of minimal $Z$-respecting budget-respecting $XY$-cuts of size $2^{\Oh(k^2 \log k)}$ that contains all important ones.
\end{theorem}

Theorem~\ref{thm:dfas} follows from Theorem~\ref{thm:imp-enum} via an analogous arguments as in~\cite{DBLP:journals/jacm/ChenLLOR08}; we postpone
them to Section~\ref{sec:imp2}. First, we focus on the proof of Theorem~\ref{thm:imp-enum}.

\begin{proof}[Proof of Theorem~\ref{thm:imp-enum}.]
Consider the recursive algorithm presented in Figure~\ref{fig2}.
The recursive procedure \textbf{ImportantCut} takes as an input a \textsc{Multi-budgeted Cut} instance
$I = (G,X,Y,\ell,(E_i,k_i)_{i=1}^\ell)$ and a set $Z \subseteq \bigcup_{i=1}^\ell E_i$, with the goal to enumerate
all important $Z$-respecting budget-respecting $XY$-cuts. Note that the procedure may output some more $Z$-respecting budget-respecting $XY$-cuts;
we need only to ensure that
\begin{enumerate}
\item it outputs all important ones,\label{i:imp:all}
\item it outputs $2^{\Oh(k^2 \log k)}$ cuts,\label{i:imp:cnt} and
\item it runs within the desired time.\label{i:imp:time}
\end{enumerate}
The procedure first invokes the algorithm of Theorem~\ref{thm:ff} on $(G,X,Y,k,Z)$, where $k = \sum_{i=1}^\ell k_i$.
If the call returned that there is no $Z$-respecting $XY$-cut of size at most $k$, we can return an empty set.
Otherwise, let $(P_j)_{j=1}^\lambda$, $B$, and $C$ be the computed objects.
We perform a branching step, with each branch labeled with a tuple $(A_1,A_2,\ldots,A_\ell)$ where $A_i \subseteq [\lambda]$ and $|A_i| \leq k_i$ for every $i \in [\ell]$.
A branch $(A_1,A_2,\ldots,A_\ell)$ is supposed to capture important cuts $C_0$ with $\{j | C_0 \cap B \cap E(P_j) \cap E_i \neq \emptyset\} \subseteq A_i$
for every $i \in [\ell]$; that is, for every $i \in [\ell]$ and $j \in [\lambda]$ we guess if $C_0$ contains a \emph{bottleneck} edge of color $i$ on path $P_j$.
All this information (i.e., paths $P_j$, the set $B$, the cut $C$, and the sets $A_i$) are passed to an auxiliary procedure \textbf{Enum}.

The procedure \textbf{Enum} shrinks the set $Z$ according to sets $A_i$.
More formally, for every $i \in [\ell]$ and $j \in [\lambda] \setminus A_i$ we delete from $Z$ all edges from $B \cap E_i \cap E(P_j)$, obtaining a set
$\widehat{Z} \subseteq Z$.
At this point, we check if the reduction of the set $Z$ to $\widehat{Z}$ increased the size of minimum $Z$-respecting $XY$-cut by invoking
Theorem~\ref{thm:ff} on $(G,X,Y,k,\widehat{Z})$ and obtaining objects $(\widehat{P}_j)_{j=1}^{\widehat{\lambda}}, \widehat{B}, \widehat{C}$
or a negative answer.
If the size of the minimum cut increased, that is, $\widehat{\lambda} > \lambda$, we recurse with the original procedure \textbf{ImportantCut}.
Otherwise, we add one cut to $\mathcal{S}$, namely $\widehat{C}$. Furthermore, we try to shrink one of the sets $A_i$ by one and recurse;
that is, for every $i \in [\ell]$ and every $j \in A_i$, we recurse with the procedure \textbf{Enum} on sets $A_{i'}'$ where
$A_i' = A_i \setminus \{j\}$ and $A_{i'}' = A_{i'}$ for every $i' \in [\ell] \setminus \{i\}$.

\begin{figure}[htbp]
\setbox4=\vbox{\hsize32pc \strut \begin{quote}
\vspace*{-5mm} \footnotesize
\textbf{ImportantCut}($I,Z$) \\
Input: A \textsc{Multi-budgeted cut} instance $I = (G,X,Y,\ell,(E_i,k_i)_{i=1}^\ell)$ and a set $Z \subseteq \bigcup_{i=1}^\ell E_i$. \\
Output: a family $\mathcal{S}$ of minimal $Z$-respecting budget-respecting $XY$-cuts that contains all important ones.

1. $\mathcal{S} := \emptyset$; \\
2. apply the algorithm of Theorem~\ref{thm:ff} to $(G,X,Y,k,Z)$ with $k = \sum_{i=1}^\ell k_i$, obtaining either
objects $(P_i)_{i=1}^\lambda$, $B$, and $C$, or an answer NO; \\
3. \textbf{if} an answer NO is obtained, \textbf{then} \textbf{return} $\mathcal{S}$;  \\
4. \textbf{for each} $(A_{1},...,A_{\ell})$ such that for every $i$ in $[\ell]$, $A_{i}\subseteq [\lambda]$ and $|A_{i}|\leq k_i$ \textbf{do} \\
4.1 \hspace*{2mm} $\mathcal{S} :=  \mathcal{S} \cup \mathbf{Enum}(I,Z,(P_j)_{j=1}^\lambda,B,C,(A_i)_{i=1}^\ell)$ \\
5. \textbf{return} $\mathcal{S}$ \\

\textbf{Enum}($I,Z,(P_j)_{j=1}^\lambda,B,C,(A_i)_{i=1}^\ell$) \\
Input: A \textsc{Multi-budgeted cut} instance $I = (G,X,Y,\ell,(E_i,k_i)_{i=1}^\ell)$, a set $Z \subseteq \bigcup_{i=1}^\ell E_i$,
 a family $(P_j)_{j=1}^\lambda$ of paths from $X$ to $Y$ such that every edge of $Z$ appears on at most one path $P_j$,
 a set $B$ consisting of all edges that participate in some minimum $Z$-respecting $XY$-cut, a minimum $Z$-respecting $XY$-cut $C$
  closest to $Y$, and sets $A_i \subseteq [\lambda]$ of size at most $k_i$ for every $i \in [\ell]$ \\
Output: a family $\mathcal{S}$ of minimal $Z$-respecting budget-respecting $XY$-cuts that contains all cuts $C_0$ that are
important $Z$-respecting budget respecting $XY$-cuts and satisfy $\{j | E(P_j) \cap B \cap C_0 \cap E_i \neq \emptyset \} \subseteq A_i$ for every $i \in [\ell]$.

a. $\widehat{Z}:=Z$;   \\
b. \textbf{for each} $i\in [\ell]$ \textbf{do}  \\
\hspace*{11mm} \textbf{for each} $j\in [\lambda]\setminus{A_{i}}$ \textbf{do} \\
\hspace*{16mm} $\widehat{Z} := \widehat{Z} \setminus (B \cap E_i \cap E(P_j))$; \\
c. apply the algorithm of Theorem~\ref{thm:ff} to $(G,X,Y,k,\widehat{Z})$, obtaining either
objects $(\widehat{P}_i)_{i=1}^{\widehat{\lambda}}$, $\widehat{B}$, and $\widehat{C}$ or an answer NO; \\
d. \textbf{if} $\widehat{\lambda}$ exists and $\widehat{\lambda}>\lambda$, \textbf{then} \\
d.1 \hspace*{2mm} $\mathcal{S} := \mathcal{S} \cup$ \textbf{ImportantCut}$(I,\widehat{Z})$; \\
e. \textbf{else} \textbf{if} $\widehat{\lambda}$ exists and equals $\lambda$, \textbf{then} \\
e.1 \hspace*{2mm} $\mathcal{S} := \mathcal{S} \cup \{\widehat{C}\}$; \\
e.2 \hspace*{2mm} \textbf{for each} $i \in [\ell]$ \textbf{do} \\
\hspace*{11mm} \textbf{for each} $j \in A_i$ \textbf{do} \\
\hspace*{16mm} $A_i' := A_i \setminus \{j\}$ and $A_{i'}' := A_{i'}$ for every $i' \in [\ell] \setminus \{i\}$ \\
\hspace*{16mm} $\mathcal{S} := \mathcal{S} \cup \mathbf{Enum}(I,\widehat{Z},(P_j)_{j=1}^\lambda,\widehat{B},\widehat{C},(A_i')_{i=1}^\ell)$.\\
f. \textbf{return} $\mathcal{S}$ \\
\end{quote}

\vspace*{-6mm}
\strut}
$$\boxit{\box4}$$
\vspace*{-9mm}
\caption{FPT algorithm for enumerating important multi-budgeted $Z$-respecting $XY$-cuts} \label{fig2}
\end{figure}

Let us first analyze the size of the search tree. A call to \textbf{ImportantCut} invokes at most $\binom{\lambda \ell}{\leq k} \leq (k\ell+1)^k$ calls
  to \textbf{Enum}. Each call to \textbf{Enum} either falls back to \textbf{ImportantCut} if $\widehat{\lambda} > \lambda$
  or branches into $\sum_{i=1}^\ell |A_i| \leq k\ell$ recursive calls to itself. In each recursive call,
the sum $\sum_{i=1}^\ell |A_i|$ decreases by one. Consequently, the initial call to \textbf{Enum} results in at most
  $(k\ell)^k$ recursive calls, each potentially falling back to \textbf{ImportantCut}.
  Since each recursive call to \textbf{ImportantCut} uses strictly larger value of $\lambda$, which cannot grow larger than $k$, and $\ell \leq k$,
the total size of the recursion tree is $2^{\Oh(k^2 \log k)}$.
Each recursive call to \textbf{Enum} adds at most one set to $\mathcal{S}$, while each recursive call to \textbf{ImportantCut}
and \textbf{Enum} runs in time $\Oh(2^{k\ell} \cdot k \cdot (|V(G)|+|E(G)|))$.
  The promised size of the family $\mathcal{S}$ and the running time bound follows.
It remains to show correctness, that is, that every important $Z$-respecting budget-respecting $XY$-cut is contained in $\mathcal{S}$ returned by a call to \textbf{ImportantCut}$(I,Z)$.

We prove by induction on the size of the recursion tree that (1) every call to \textbf{ImportantCut}$(I,Z)$ enumerates all important $Z$-respecting
budget-respecting $XY$-cuts, and
(2) every call to \textbf{Enum}$(I,Z,(P_j)_{j=1}^\lambda,B,C,(A_i)_{i=1}^\ell)$ enumerates all important
$Z$-respecting budget-respecting $XY$-cuts $C_0$ with the property that $\{j | E_i \cap E(P_j) \cap B \cap C_0 \neq \emptyset\} \subseteq A_i$ for every $i \in [\ell]$.

The inductive step for a call \textbf{ImportantCut}$(I,Z)$ is straightforward. Let us fix an arbitrary important $Z$-respecting budget-respecting $XY$-cut $C_0$.
Since $C_0$ is budget-respecting, $C_0$ is a $Z$-respecting cut of size at most $k$, and thus the initial call to Theorem~\ref{thm:ff} cannot return NO.
Consider the tuple $(A_1,A_2,\ldots,A_\ell)$ where for every $i\in [\ell]$, $\{j|E(P_{j})\cap E_{i}\cap B\cap C_0\}=A_{i}$.
Since $C_0$ is budget-respecting and the paths $P_j$ do not share an edge of $Z$, we have that $|A_i| \leq k_i$ for every $i \in [\ell]$ and the
algorithm considers this tuple in one of the branches.
Then, from the inductive hypothesis, the corresponding call to \textbf{Enum} returns a set containing $C_0$.

Consider now a call to \textbf{Enum}$(I,Z,(P_j)_{j=1}^\lambda,B,C,(A_i)_{i=1}^\ell)$ and an important
$Z$-respecting budget-respecting $XY$-cuts $C_0$ with the property that $\{j | E_i \cap E(P_j) \cap B \cap C_0 \neq \emptyset\} \subseteq A_i$ for every $i \in [\ell]$.
By the construction of $\widehat{Z}$ and the above assumption, $C_0$ is $\widehat{Z}$-respecting.
In particular, the call to the algorithm of Theorem~\ref{thm:ff} cannot return NO.
Hence, in the case when $\widehat{\lambda} > \lambda$, $C_0$ is enumerated by the recursive call to \textbf{ImportantCut} and we are done.
Assume then $\widehat{\lambda} = \lambda$.

For $i \in [\ell]$, let $\widehat{A}_i = \{ j | E_i \cap E(P_j) \cap \widehat{B} \cap C_0 \neq \emptyset\}$. Since $\widehat{Z} \subseteq Z$ but the sizes
of minimum $Z$-respecting and $\widehat{Z}$-respecting $XY$-cuts are the same, we have $\widehat{B} \subseteq B$.
Consequently, $\widehat{A}_i \subseteq A_i$ for every $i \in [\ell]$.

Assume there exists $i \in [\ell]$ such that $\widehat{A}_i \subsetneq A_i$ and let $j \in A_i \setminus \widehat{A}_i$.
Consider then the branch $(i,j$) of the \textbf{Enum} procedure, that is, the recursive call with $A_i' = A_i \setminus \{j\}$
and $A_{i'}' = A_{i'}$ for $i' \in [\ell] \setminus \{i\}$.
 Observe that we have $\{j | E_{i'} \cap E(P_j) \cap \widehat{B} \cap C_0 \neq \emptyset\} \subseteq A_{i'}'$
for every $i' \in [\ell]$ and, by the inductive hypothesis, the corresponding call to \textbf{Enum} enumerates $C_0$.
Hence, we are left only with the case $\widehat{A}_i = A_i$, that is, $A_i = \{j | E_i \cap E(P_j) \cap \widehat{B} \cap C_0 \neq \emptyset\}$ for every $i \in [\ell]$.

We claim that in this case $C_0 = \widehat{C}$. Assume otherwise.
Since $|\widehat{C}| = \widehat{\lambda}=\lambda$ and $\widehat{Z} \subseteq Z$, $\widehat{C}$ contains exactly one edge on every path $P_j$.
Also, $\widehat{C} \subseteq \widehat{B}$ by the definition of the set $\widehat{B}$. Since $\widehat{C}$ is the minimum $\widehat{Z}$-respecting
$XY$-cut that is closest to $Y$, $\widehat{C} = \{e_1,e_2,\ldots,e_\lambda\}$ where $e_j$ is the last (closest to $Y$) edge of $\widehat{B}$
on the path $P_j$ for every $j \in [\lambda]$.

Let $R_0$ and $\widehat{R}$ be the set of vertices reachable from $X$ in $G-C_0$ and $G-\widehat{C}$, respectively.
Let $D$ be a minimal $XY$-cut contained in $\delta^{+}(R_0 \cup \widehat{R})$.
(Note that $\delta^{+}(R_0 \cup \widehat{R})$ is an $XY$-cut because $X\subseteq R_0\cup \widehat{R}$ and $Y\cap (R_0\cup \widehat{R})=\emptyset$.)
Then since $D \subseteq C_0\cup \widehat{C}\subseteq Z$, $D$ is $Z$-respecting.
By definition, every vertex reachable from $X$ in $G-R_0$ is also reachable from $X$ in $G-D$.

We claim that $D$ is budget-respecting and, furthermore, dominates $C_0$.
Fix a color $i \in [\ell]$; our goal is to prove that $|D \cap E_i| \leq |C_0 \cap E_i|$.
To this end, we charge every edge of color $i$ in $D \setminus C_0$ to a distinct edge of color $i$ in $C_0 \setminus D$.
Since $D \subseteq C_0 \cup \widehat{C}$, we have that $D \setminus C_0 \subseteq \widehat{C}$, that is,
      an edge of $D \setminus C_0$ of color $i$ is an edge $e_j$ for some $j \in [\lambda]$ with $e_j \in E_i$ and $e_j \in D \setminus C_0$.

Recall that we are working in the case $A_i = \{j | E_i \cap E(P_j) \cap \widehat{B} \cap C_0 \neq \emptyset\}$.
Since $e_j \in \widehat{C} \subseteq \widehat{Z}$, we have that $j \in A_i$. Hence, there exists $e_j' \in E_i \cap E(P_j) \cap \widehat{B} \cap C_0$.
By the definition of $\widehat{C}$, $e_j$ is the last (closest to $Y$) edge of $\widehat{B}$ on $P_j$.
Since $e_j \notin C_0$, $e_j' \neq e_j$ and $e_j'$ lies on the subpath of $P_j$ between $X$ and the tail of $e_j$.
This entire subpath is contained in $\widehat{R}$ and, hence, $e_j' \notin D$.

We charge $e_j$ to $e_j'$. Since $e_j' \in E(P_j) \cap E_i \cap \widehat{B} \cap (C_0 \setminus D)$, for distinct $j$, the edges $e_j'$ are distinct as the paths $P_j$ do not share an edge belonging to $Z$ and $\widehat{B} \subseteq \widehat{Z} \subseteq Z$.
Consequently, $|D \cap E_i| \leq |C_0 \cap E_i|$.
This finishes the proof that $D$ dominates $C_0$.

Since $C_0$ is important, we have $D = C_0$. In particular, $\widehat{R} \subseteq R_0$.
On the other hand, for every $j \in [\lambda]$ we have that $e_j \in \widehat{C} \subseteq \widehat{Z} \subseteq Z \subseteq \bigcup_{i=1}^\ell E_i$.
In particular, there exists $i \in [\ell]$ such that
$e_j \in E_i$ and $j \in A_i$. Hence, we also have $E_i \cap E(P_j) \cap \widehat{B} \cap C_0 \neq \emptyset$. But the entire subpath of $P_j$ from $X$
to the tail of $e_j$ lies in $\widehat{R} \subseteq R_0$, while $e_j$ is the last edge of $\widehat{B}$ on $P_j$. Hence, $e_j \in C_0$.
Since the choice of $j$ is arbitrary, $\widehat{C} \subseteq C_0$. Since $\widehat{C}$ is an $XY$-cut and $C_0$ is minimal, $\widehat{C} = C_0$ as claimed.

This finishes the proof of Theorem~\ref{thm:imp-enum}.
\end{proof}

\subsection{Applications}\label{sec:imp2}

The \textsc{Directed Feedback Arc Set} problem is a classic problem that played major role in the development of parameterized complexity.
In this problem, given a directed graph $G$ and an integer $k$, the problem is to decide if there exists an arc set $S$ of size at most $k$ such that $G-S$ has no directed cycles.
The multi-budgeted variant is defined as follows.

\begin{quote}
\textsc{Multi-budgeted Directed Feedback Arc Set}\\
\textbf{Input:} A directed graph $G$, an integer $\ell$, and for every $i \in \{1,2,\ldots,\ell\}$ a set $E_i \subseteq E(G)$ and an integer $k_i \geq 1$.\\
\textbf{Question:} Is there a set of arcs $S\subseteq \bigcup_{i=1}^{\ell}E_{i}$ such that there is no directed cycle in $G-S$ and for every $i\in [\ell]$, $|S\cap E_i|\leq k_{i}$.\\
\end{quote}

The first FPT algorithm for the \textsc{Directed Feedback Arc Set} problem is given by Chen et al.~\cite{DBLP:journals/jacm/ChenLLOR08}.
  In their algorithm, they use iterative compression and  reduce the \textsc{Directed Feedback Arc Set} compression problem to the \textsc{Skew Edge Multicut} problem.
They propose a pushing lemma for \textsc{Skew Edge Multicut} and solve \textsc{Skew Edge Multicut} through enumerating important cuts.
We show that for the multi-budgeted variant, a similar strategy works with the help of Theorem~\ref{thm:imp-enum}.
Formally, the \textsc{Multi-budgeted Skew Edge Multicut} problem is defined as follows.

\begin{quote}
\textsc{Multi-budgeted Skew Edge Multicut}\\
\textbf{Input:} A directed graph $G$, an integer $\ell$, for every $i \in \{1,2,\ldots,\ell\}$ a set $E_i \subseteq E(G)$ and an integer $k_i \geq 1$, and a sequence $(s_i,t_i)_{i=1}^q$ of terminal pairs.\\
\textbf{Question:} Is there a set of arcs $C\subseteq \bigcup_{i=1}^{\ell}E_{i}$ such that there is no directed path from $s_{i}$ to $t_{j}$ for any $i\geq j$ in $G-C$ and for every $i\in [\ell]$, $|C\cap E(i)|\leq k_{i}$?\\
\end{quote}

As in the case of \textsc{Multi-budgeted Cut}, the assumption that $k_i \geq 1$ for every $i \in [\ell]$ implies $\ell \leq k$.

We start by observing a direct corollary of the maximality criterium in the definition of important budget-respecting separators.

\begin{lemma}\label{existing}
Given an instance $(G,X,Y,\ell,(E_i,k_i)_{i=1}^\ell)$ of \textsc{Multi-budgeted cut}, for every minimal budget-respecting $XY$-cut $C$ there exists an important
  budget-respecting $XY$-cut $C'$ that dominates $C$.
\end{lemma}

Similar to the pushing lemma for \textsc{Skew Edge Multicut}~\cite{DBLP:journals/jacm/ChenLLOR08}, we propose a pushing lemma for the multi-budgeted variant.
\begin{lemma}\label{pushing-lemma}
Every yes-instance $I=(G,\ell,(E_i,k_i)_{i=1}^\ell,(s_i,t_i)_{i=1}^q)$ of \textsc{Skew Edge Multicut} admits a solution
  that contains an important budget-respecting $XY$-cut for $X = \{s_q\}$ and $Y=\{t_1,t_2,\ldots,t_q\}$.
\end{lemma}
\begin{proof}
Let $C$ be a solution to $I$.
Let $X = \{s_q\}$, $Y=\{t_{1},...,t_{q}\}$, and $R$ be the set of vertices reachable from $s_{q}$ in $G-C$.
Since $C$ is a solution, $\delta^{+}(R)\subseteq C$ is a budget-respecting $XY$-cut; let $D \subseteq \delta^{+}(R)$ be a minimal one.
By Lemma~\ref{existing}, there exists an important budget-respecting $XY$-cut $D^*$ dominating $D$.
Let $R^*$ be the set of vertices reachable from $s_q$ in $G-D^*$.
We claim that $C^* := (C \setminus D) \cup D^*$ is a solution to $I$ as well.

Suppose for contradiction that there is a directed path $P$ from $s_{i}$ to $t_{j}$ for some $i\geq j$ in $G-C^{*}$. If $P$ contains one vertex of $R^*$, it contradicts that $D^{*}$ is an $XY$-cut
because $P$ must contain one edge of $D^*$. Thus $P$ is disjoint from $R^*$. Since $R\subseteq R^*$, $P$ is disjoint from $R$, and hence $P$ is disjoint from $D$.
Since $P$ is not cut by $C^{*}=(C\setminus D)\cup D^*$, $P$ is not cut by $C\setminus D$. It follows that $P$ is not cut by $C=(C\setminus D)\cup D$, contradicting that $C$ is a solution.

To complete the proof, note that for every $i \leq [\ell]$ we have $|D^* \cap E_i| \leq |D \cap E_i|$ since $D'$ dominates $D$, and hence $|C^* \cap E_i| \leq |C \cap E_i|$.
Consequently, $C^*$ is budget-respecting.
\end{proof}

Lemma~\ref{pushing-lemma} yields the following branching strategy.
\begin{lemma}\label{lem:skew}
There is an FPT algorithm for \textsc{Multi-budgeted Skew Edge Multicut} running in time $2^{\Oh(k^3 \log k))} \cdot (|V(G)| + |E(G)|)$.
\end{lemma}
\begin{proof}
We perform a recursive branching algorithm where the budgets $k_i$ will decrease, thus we allow instances with the inequalities $k_i \geq 1$ violated.
If $k_i < 0$ for some $i \in [\ell]$, then we can answer NO. Otherwise, if $q = 0$, then we can answer YES.
Otherwise, perform a depth-first search from $s_q$. If no terminal $t_i$ has been reached, delete the visited vertices (as they are not contained in any $s_j$-to-$t_{j'}$ path for any $1 \leq j,j' \leq q$)
together with $t_q$ (as there is no $s_q$-to-$t_q$ path in the graph),
decrease $q$ by one and restart the algorithm.
Since this operation can be performed in time linear in the size of the deleted part of the graph, in total it takes linear time.

Otherwise, proceed as follows. By Lemma~\ref{pushing-lemma} if the input instance is a yes-instance, there is a solution $C^{*}$ which contains an important budget-respecting
$s_q\{t_1,t_2,\ldots,t_q\}$-cut.
By Theorem~\ref{thm:imp-enum}, we can enumerate in time $2^{\Oh(k^2 \log k)} (|V(G)|+|E(G)|)$ a set of minimal budget-respecting $XY$-cuts $\mathcal{S}$
of size $2^{\Oh(k^2 \log k)}$ that contains all important ones.
We invoke this enumeration, and branch on the choice of important budget-respecting $s_q\{t_1,t_2,\ldots,t_q\}$-cut contained in the sought solution.
In a branch where a cut $D$ is chosen, we delete $D$ from the graph and decrease each budget $k_i$ by $|D \cap E_i|$.
Since at least one terminal $t_i$ is reachable from $s_q$, in every branch the cut $D$ is nonempty and thus $k = \sum_{i=1}^\ell k_i$ decreases by at least one.
Consequently, the depth of the recursion is bounded by $k$. The running time bound follows.
\end{proof}

We now use the algorithm of Lemma~\ref{lem:skew} to give an algorithm for \textsc{Multi-budgeted Directed Feedback Arc Set},
   completing the proof of Theorem~\ref{thm:dfas}.
\begin{lemma}
\textsc{Multi-budgeted Directed Feedback Arc Set} can be solved in time $2^{\Oh(k^3 \log k)} (|V(G)|+|E(G)|)$.
\end{lemma}
\begin{proof}
Let $I = (G,\ell,(E_i,k_i)_{i=1}^\ell)$ be an input instance.
  We start by applying the algorithm of Lokshtanov et al.~\cite{LokshtanovRS16} for the classic \textsc{Directed Feedback Vertex Set} on $G$
  with parameter $k = \sum_{i=1}^\ell$. If the call returned that there is no solution, we can safely return NO.
  Otherwise, let $W$ be the returned solution: $|W| \leq k$ and $G-W$ is acyclic.

Suppose $I$ is a yes-instance and there is a solution $S$.
Then $G-S$ is a directed acyclic graph, admitting a topological ordering of $V(G)$.
This ordering indices a permutation of the vertices in $W$.
In our algorithm, we branch on every permutation of the vertices in $W$, ensuring that at least one of the permutation is the same as the permutation induced by the topological ordering of $G-S$.
Let $w_{1},...,w_{|W|}$ be an arbitrary permutation of the vertices in $W$.
We construct a graph $G'$ as follows. For each $i\in [|W|]$, we replace every vertex $w_{i}$ with two vertices $s_{i},t_{i}$, every edge $(w_{i},a)$ with $(s_{i},a)$ of the same color and
every edge $(b,w_{i})$ with $(b,t_{i})$ of the same color.
Then we add a directed edge $(t_{i},s_{i})$ for each $i\in [|W|]$ with no color.
In this manner, we construct a \textsc{Multi-budgeted Skew Edge Multicut} instance $I' = (G',\ell,(E_i',k_i)_{i=1}^\ell,(s_i,t_i)_{i=1}^{|W|})$ corresponding to the permutation $w_{1},...,w_{|W|}$.

We claim that the input instance $I$ of \textsc{Multi-budgeted Directed Feedback Arc Set} is a yes-instance if and only if there exists one permutation $w_{1},...,w_{|W|}$ of $W$ such that
the corresponding \textsc{Multi-budgeted Skew Edge Multicut} instance $I'$ is a yes-instance.
For the "only if direction", let $S$ be a solution to $I$.
We have a topological ordering of $V(G)$, inducing an ordering $w_{1},...,w_{|W|}$ on $W$.
For this ordering, let $I'$ be the corresponding instance of \textsc{Multi-budgeted Skew Edge Multicut}.
According to the way we construct $G'$, every edge in $S$ has a corresponding edge in $G'$. Let $S'$ be the set of the corresponding edges in $G'$.
We claim that $S'$ is a solution for $I'$. Obviously $S'$ is budget-respecting.
Suppose for contradiction that there is a directed path $P$ from $s_{i}$ to $t_{j}$ for some $i\geq j$ in $G'-S'$.
If $i=j$, then the corresponding edges of $P$ form a directed cycle passing through $w_{i}$ in $G-S$, a contradiction.
Suppose that $i>j$. If $P$ goes through some edge in $\{(t_{i},s_{i})|i\in [|W|]\}$, then there must be a subpath of $P'$ from $s_{i'}$ to $t_{j'}$ such that $i'>j'$
and $P'$ contains no edges in $\{(t_{i},s_{i})|i\in [|W|]\}$.
Then the corresponding edges of $P'$ is a directed path from $w_{i'}$ to $w_{j'}$, contradicting that $w_{i'}$ is later than $w_{j'}$ in the topological ordering of $V(G)$ after removing $S$.

For the "if direction", suppose that $S'$ is a solution for $I'$ and $w_{1},...,w_{|W|}$ is the corresponding ordering of $W$.
Let $S$ be the set of edges consisting of the corresponding edges of $S'$. We claim that $S$ is the solution for $I$.
Obviously $S$ is budget-respecting.
Suppose that there is a cycle $Q$ in $G-S$. Since $W$ is a feedback vertex set for $G$, $Q$ must go through at least one vertex in $W$.
Suppose that $Q$ goes through a vertex in $W$, namely $w_{i}$. Then we can find a path from $s_{i}$ to $t_{i}$ in $G'-S'$,
        contradicting that $S'$ is a solution to $I'$.

This finishes the proof of the lemma and of the whole Theorem~\ref{thm:dfas}.
\end{proof}

\section{Bound on the number of solutions closest to $t$}\label{sec:bound}

\newcommand{\maze}{\mathcal{U}}

In this section we prove Theorems~\ref{thm:bnd-weights} and~\ref{thm:bnd-chain}.
The central definition of this section is the following (see also Figure~\ref{fig:bowtie}).
\begin{definition}
Let $G$ be a directed graph with distinguished vertices $s$ and $t$ and let $k$ be an integer.
An $a$-\emph{bowtie} is a sequence $C_1,C_2,\ldots,C_a$ of pairwise disjoint minimal $st$-cuts of size $k$ each such that each cut $C_i$ can be partitioned $C_i = A_i \uplus B_i$ such that
  for every $1 \leq i < j \leq a$, the set $A_i$ is exactly the set of edges of $C_i$ reachable from $s$ in $G-C_j$ and $B_j$ is exactly the set of edges of $C_j$ reachable from $s$ in $G-C_i$.
\end{definition}
\begin{figure}[htbp]
\begin{center}
\includegraphics[width=.6\linewidth]{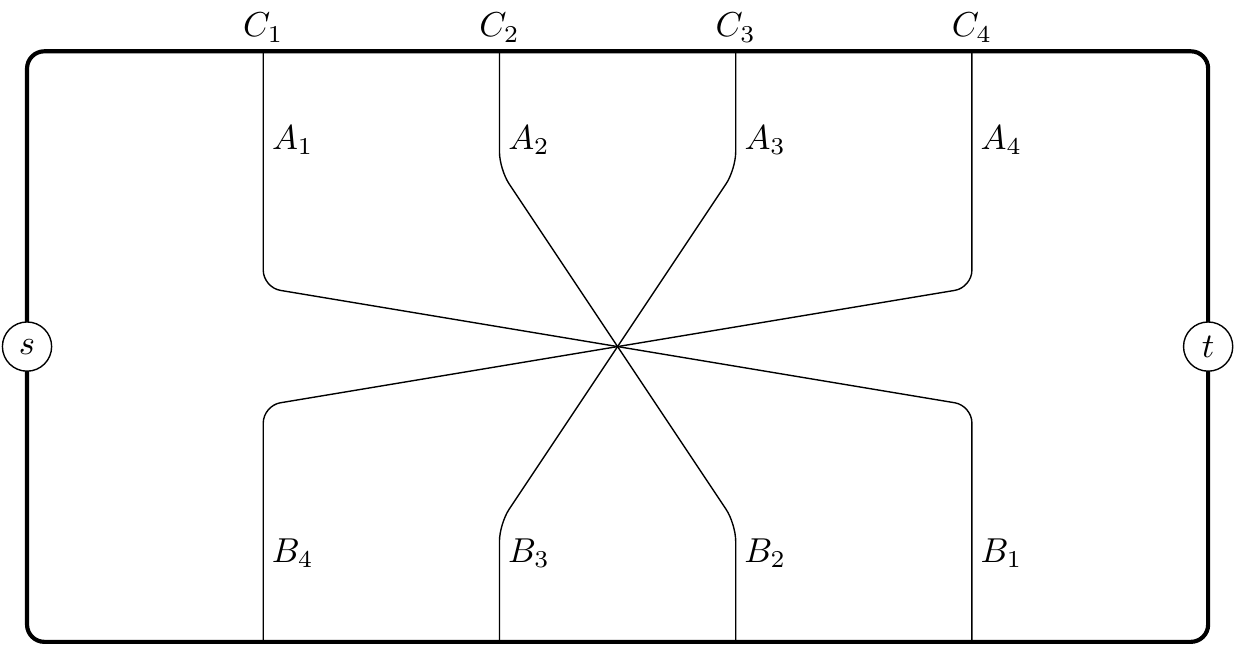}
\end{center}
\caption{A schematic picture of a $4$-bowtie.}\label{fig:bowtie}
\end{figure}

Our main graph-theoretic result is the following.
\begin{theorem}\label{thm:bowtie-cuts}
For every integers $a,k \geq 1$ there exists an integer $g$ such that for every
directed graph $G$ with distinguished $s,t \in V(G)$, and a family $\maze$ of pairwise disjoint minimal $st$-cuts of size $k$ each, if $|\maze| \geq g$, then $\maze$ contains an $a$-bowtie.
\end{theorem}

The next two lemmata are key observations to
  prove Theorems~\ref{thm:bnd-weights} and~\ref{thm:bnd-chain}, respectively,
        with the help of Theorem~\ref{thm:bowtie-cuts}.
\begin{lemma}\label{lem:bowtie-weighted}
Let $k$, $g$, $G$, $s$, $t$, $\mathcal{F}$, and $\mathcal{G}$ be as in the statement
  of Theorem~\ref{thm:bnd-weights}.
Then $\mathcal{G}$ does not contain an $a$-bowtie for $a > \binom{k+2}{2}$.
\end{lemma}
\begin{proof}
Assume the contrary, let $(C_i,A_i,B_i)_{i=1}^a$ be such a bowtie.
Since $a > \binom{k+2}{2}$, there exists $i < j$ with $|A_i| = |A_j|$ and $|B_i| = |B_j|$
(there are $\binom{k+2}{2}$ choices for $(|A_i|, |B_i|)$).
However, then $A_i \cup B_j$ and $A_j \cup B_i$ have also cardinality $k$, are $st$-cuts, and have together twice the minimum weight.
Furthermore, the set of vertices reachable from $s$ in $G-(A_j \cup B_i)$
is a strict superset of the set of vertices reachable from $s$ in $G-C_i$ and $G-C_j$.
This contradicts the fact that $C_i,C_j \in \mathcal{G}$.
\end{proof}

\begin{lemma}\label{lem:bowtie-chain}
Let $k$, $\ell$, $I=(G,s,t,(P_i)_{i=1}^m,k)$, $\mathcal{F}$, and $\mathcal{G}$
  be as in the statement
  of Theorem~\ref{thm:bnd-chain}.
Then $\mathcal{G}$ does not contain a $4$-bowtie $(C_i,A_i,B_i)_{i=1}^4$
in which the edge set of every path $P_j$ intersects at most one cut $C_i$.
\end{lemma}
\begin{proof}
Assume the contrary. Let $(C_i, A_i, B_i)_{i=1}^4$ be such a $4$-bowtie.
Consider $i \in \{2,3\}$ and two edges $e \in A_i$ and $f \in B_i$.
In $G-C_4$, the edge $e$ is reachable from $s$ while $f$ is not; consequently, $e$ and $f$ cannot appear on the same input path with $e$ being earlier
(by assumption, $C_4$ is disjoint from the input path in question).
A similar reasoning for $G-C_1$ shows that $e$ and $f$ cannot appear on the same input path with $f$ being earlier than $e$.

Hence, $e$ and $f$ cannot appear together on a single path $P_j$.
For a set of edges $D$, by the \emph{cost} of $D$ we denote $|\{j | D \cap P_j \neq \emptyset\}|$.
Since the choice of $e$ and $f$ was arbitrary, we infer that
the sum of costs of $A_2 \cup B_3$ and of $A_3 \cup B_2$ equals the sum of costs of $C_2$ and of $C_3$.
Hence, both these $st$-cuts have minimum cost.
However, $A_2 \cup B_3$ is closer to $t$ than $C_2$, a contradiction.
\end{proof}

\begin{proof}[Proof of Theorem~\ref{thm:bnd-weights}.]
  Assume $|\mathcal{G}| > g$ for some sufficiently large $g$ to be fixed later.
  For $i \in [k]$, let $\mathcal{G}^i$ be the set of $u \in \mathcal{G}$ of cardinality $i$.
  We apply the Sunflower Lemma to the largest set $\mathcal{G}^i$:
  If $g > k \cdot k! g_1^k$ for some integer $g_1$ to be chosen later, there exists
   $\mathcal{G}_1 \subseteq \mathcal{G}$ with $|\mathcal{G}_1| > g_1$, every element of $\mathcal{G}_1$ being of the same size $k'$, and a set $c$
     such that $u \cap v = c$ for every distinct $u,v \in \mathcal{G}_1$.

 Let $\widehat{k} = k' - |c|$, $\widehat{u} = u \setminus c$ for every $u \in \mathcal{G}_1$,
 $\widehat{\mathcal{G}}_1 = \{\widehat{u}~|~u \in \mathcal{G}_1\}$ and $\widehat{G} = G-c$.
 Since every $u \in \maze$ is a minimal $st$-cut of size $k'$ in $G$,
 every $\widehat{u} \in \widehat{\mathcal{G}}_1$ is a minimal $st$-cut of size $\widehat{k}$ in $\widehat{G}$.
   Furthermore, every $\widehat{u} \in \widehat{\mathcal{G}}_1$ is a minimal $st$-cut of size at most $\widehat{k}$ in $\widehat{G}$ of minimum possible weight:
   if there existed an $st$-cut $\widehat{x}$ of smaller weight and cardinality at most $\widehat{k}$, then
   $x = \widehat{x} \cup c$ would be an $st$-cut in $G$ of cardinality at most $k$ and weight smaller than every element of $\mathcal{G}_1$.
Similarly, if there were a minimal $st$-cut $\widehat{x}$ in $\widehat{G}$ of minimum weight and cardinality at most $\widehat{k}$
that is closer to $t$ than $\widehat{u}$ for some $\widehat{u} \in \widehat{\mathcal{G}}_1$, then $\widehat{x} \cup c$ would be an
$st$-cut in $G$ of cardinality at most $k$ and minimum weight that is closer to $t$ than $u$, a contradiction.
By construction, the elements of $\widehat{\mathcal{G}}_1$ are pairwise disjoint.

Lemma~\ref{lem:bowtie-weighted} bounds the maximum possible size of a bowtie in $\widehat{\mathcal{G}}_1$.
Hence, Theorem~\ref{thm:bowtie-cuts} asserts that $\widehat{\mathcal{G}}_1$ has size bounded by a function of $k$.
This finishes the proof of the theorem.
\end{proof}

\begin{proof}[Proof of Theorem~\ref{thm:bnd-chain}.]
We proceed similarly as in the proof of Theorem~\ref{thm:bnd-weights}, but we need to be a bit more careful
with the paths $P_j$. Assume $|\mathcal{G}| > g$ for some sufficiently large integer $g$.

As before, we partition $\mathcal{G}$ according to the sizes of elements: for every $i \in [k\ell]$, let $\mathcal{G}^i = \{u \in \mathcal{G}~|~|u|=i\}$.
Let $i \in [k\ell]$ be such that $|\mathcal{G}^i| > g/(k\ell)$.
For $u \in \mathcal{G}^i$, let $J(u) = \{j~|~u \cap P_j \neq \emptyset\}$.
By the assumptions of the theorem, every set $J(u)$ is of cardinality exactly $k$.
We apply the Sunflower Lemma to $\{J(u)~|~u \in \mathcal{G}^i\}$:
If $g > (k\ell) \cdot k! \cdot g_1^k$ for some integer $g_1$ to be fixed later,
then there exists $\mathcal{G}_1 \subseteq \mathcal{G}^i$ of size larger than $g_1$ and a set $I \subseteq [m]$
such that for every distinct $u,v \in \mathcal{G}_1$ we have $J(u) \cap J(v) = I$.
For every $u \in \mathcal{G}_1$, let $u_I = u \cap \bigcup_{j \in I} P_j$.
Since $|I| \leq k$, there are at most $2^{k\ell}$ choices for $u_I$ among elements $u \in \mathcal{G}_1$.
Consequently, there exists $\mathcal{G}_2 \subseteq \mathcal{G}_1$ of cardinality larger than $g_2 := g_1 / 2^{k\ell}$
such that $u_I = v_I$ for every $u,v \in \mathcal{G}_2$. Denote $c = u_I$ for any $u \in \mathcal{G}_2$.

Let $\widehat{u} := u-c$ for every $u \in \mathcal{G}_2$.
Let $\widehat{\mathcal{G}}_2 = \{\widehat{u}~|~u \in \mathcal{G}_2\}$.

Define now $\widehat{G} = G-c$ and define a partition $\widehat{\mathcal{P}}$ of $E(\widehat{G})$ into paths of length at most $\ell$
as follows: we take all paths $P_i$ for $i \notin I$ and, for every $i \in I$, each edge of $P_i \setminus c$ as a length-$1$ path.
Furthermore, denote $\widehat{k} = k - |I|$.
Note that $(\widehat{G},s,t,\widehat{\mathcal{P}},\widehat{k})$ is a \textsc{Chain $\ell$-SAT} instance for
which every $\widehat{u} \in \widehat{\mathcal{G}}_2$ is a solution.
Furthermore, $(\widehat{G},s,t,\widehat{\mathcal{P}},\widehat{k}-1)$ is a no-instance, as if
$\widehat{x}$ were its solution, then $\widehat{x} \cup c$ would be a solution to $(G,s,t,(P_i)_{i=1}^m,k-1)$, a contradiction.
Similarly, if there were a solution $\widehat{x}$ to $(\widehat{G},s,t,\widehat{\mathcal{P}},\widehat{k})$
that is closer to $t$ than $\widehat{u}$ for some $\widehat{u} \in \widehat{\mathcal{G}}_2$, then $\widehat{x} \cup c$ would be a solution
to $(G,s,t,(P_i)_{i=1}^m,k)$ that is closer to $t$ than $u$, a contradiction.
Furthermore, by construction, the elements of $\widehat{\mathcal{G}}_2$ are pairwise disjoint and no path of $\widehat{\mathcal{P}}$ intersects
more than one element of $\widehat{\mathcal{G}}_2$.

Lemma~\ref{lem:bowtie-chain} bounds the maximum possible size of a bowtie in $\widehat{\mathcal{G}}_2$.
Hence, Theorem~\ref{thm:bowtie-cuts} asserts that $\widehat{\mathcal{G}}_2$ has size bounded by a function of $k$ and $\ell$.
This finishes the proof of the theorem.
\end{proof}

We now focus on the proof of Theorem~\ref{thm:bowtie-cuts}.
To this end, we need to introduce some abstract notions.

Let $k$ be an integer.
A \emph{$k$-maze} is a family $\maze$ of pairwise disjoint sets of size $k$,
together with a function $f_{u,v} : u \to \{\bot,\top\}$
for every ordered pair $u,v \in \maze$, $u \neq v$.
A \emph{flower} in a $k$-maze $\maze$ is a subset $\mathcal{F} \subseteq \maze$ such that there exists a value $\zeta \in \{\bot,\top\}$ and an element $e(u) \in u$ for every $u \in \mathcal{F}$
such that $f_{u,v}(e(u)) = \zeta$ for every $u,v \in \mathcal{F}$, $u \neq v$.
An \emph{$a$-bowtie} in a $k$-maze $\maze$ is a sequence $u_1,u_2,\ldots,u_a$ of pairwise distinct elements of $\maze$ together with a partition $u_i = a_i \uplus b_i$ of every set in the sequence, such that
  the following holds for every $1 \leq i,j \leq a$, $i \neq j$:
  \begin{itemize}
  \item if $i < j $ then $f_{u_i, u_j}(e) = \bot$ for $e \in a_i$ and $f_{u_i, u_j}(e) = \top$ for $e \in b_i$;
  \item if $i > j $ then $f_{u_i, u_j}(e) = \top$ for $e \in a_i$ and $f_{u_i, u_j}(e) = \bot$ for $e \in b_i$.
  \end{itemize}

We need two basic observations.
First, flowers in restrictions project to flowers in original mazes. That is, given a $k$-maze $\maze$ and a set $\hat{u} \subseteq u$ for every $u \in \maze$
with $|\hat{u}| = \hat{k}$, the natural restrictions of the functions $f_{u,v}$ give a structure of a $\hat{k}$-maze on $\hat{\maze} := \{\hat{u} : u \in \maze\}$.
It is immediate from the definition that a flower in $\hat{\maze}$ projects back to a flower in $\maze$.
Second, a reverse of a bowtie is a bowtie as well, but one needs
to swap the roles of $a_i$ and $b_i$. That is, one can check directly from the definition that if
$(u_i,a_i,b_i)_{i=1}^a$ is an $a$-bowtie, then $(u_i,b_i,a_i)_{i=a}^1$ is an $a$-bowtie as well.

An iterated Ramsey argument shows the following.
\begin{theorem}\label{thm:maze}
For every integers $k,a,b \geq 1$ there exists an integer $g$ such that any $k$-maze of size at least $g$
  contains either an $a$-bowtie or a flower of size at least $b$.
\end{theorem}
\begin{proof}
We prove Theorem~\ref{thm:maze} by
induction over $k$ and denote the resulting value $g$ as $g(k, a, b)$.
Pick $k,a,b \geq 1$ and pick a $k$-maze $\maze$.
For every $u \in \maze$ pick an arbitrary element $e_u \in u$ and denote
$\hat{u} = u \setminus \{e_u\}$.

Fix an arbitrary total order $\prec$ on $\maze$.
Consider a complete graph $H$ with vertex set $\maze$ and edges with the following annotations for distinct $u,v \in \maze$ with $u \prec v$.
\begin{itemize}
\item An edge $uv$ is \emph{light red} if $f_{u,v}(e_u) = f_{v,u}(e_v) = \bot$.
\item An edge $uv$ is \emph{dark red} if $f_{u,v}(e_u) = f_{v,u}(e_v) = \top$.
\item An edge $uv$ is \emph{light blue} if $f_{u,v}(e_u) = \bot$ and $f_{v,u}(e_v) = \top$.
\item An edge $uv$ is \emph{dark blue} if $f_{u,v}(e_u) = \top$ and $f_{v,u}(e_v) = \bot$.
\end{itemize}
We say that an edge $uv$ is \emph{red} if it is light or dark red, and similarly for blue.

By Ramsey's theorem, if $\maze$ is large enough, we have one of the following objects.
\begin{itemize}
\item A light red clique of size $b$ or a dark red clique of size $b$.
  But such a clique is in fact a flower of size $b$
  with $e(u) = e_u$ and $\zeta = \bot$ if it is a light red clique or $\zeta= \top$
  if it is a dark red clique.
\item A light blue clique or a dark blue clique of size $m_1$, for some integer $m_1$ to be
  fixed later. We denote this clique by $\mathcal{C}$ and proceed further.
\end{itemize}
If $k=1$, we take $m_1 = a$ and observe that such a blue clique is an $a$-bowtie.

For $k > 1$, we define a $(k-1)$-maze $\hat{\maze} := \{\hat{u} : u \in \mathcal{C}\}$.
We take $m_1 = g(k-1,m_2,b)$ where $m_2 = (a-1)^2 + 1$
and apply the inductive assumption to $\hat{\maze}$, obtaining either an $m_2$-bowtie
or a $b$-flower.
If the inductive assumption gives a flower, it projects back to a flower in $\maze$.
Otherwise, we have an $m_2$-bowtie
$(\hat{u}_i, \hat{a}_i, \hat{b}_i)_{i=1}^{m_2}$.
Since $m_2 = (a-1)^2 + 1$, there exist indices $1 \leq i_1 < i_2 < ... < i_a \leq m_2$ such that
$(u_{i_j})_{j=1}^a$ is ordered increasingly or decreasingly by $\prec$.
Then, by the definition of light and dark blue edges,
  either
$(u_{i_j}, \hat{a}_{i_j} \cup \{e_{u_{i_j}}\}, \hat{b}_{i_j})_{j=1}^a$
or
$(u_{i_j}, \hat{a}_{i_j}, \hat{b}_{i_j} \cup \{e_{u_{i_j}}\})_{j=1}^a$
is an $a$-bowtie. This finishes the proof.
\end{proof}

Let us now relate the above abstract notions and results
to cuts of cardinality at most $k$ in a directed
graph $G$ with distinguished vertices $G$, $s$, and $t$.
For two disjoint minimal $st$-cuts $C, D$ of size $k$,
we can define $f_{C,D} : C \to \{\bot,\top\}$ as follows: $f_{C,D}(e) = \bot$ if the tail of $e$ is reachable from $s$ in $G-D$, and $f_{C,D}(e) = \top$ otherwise.
This definition gives a structure of a $k$-maze on a family $\maze$
of pairwise disjoint minimal $st$-cuts of size $k$.
Furthermore, a direct check from the definitions show that the two notions of a \emph{bowtie}
are equivalent.
By Theorem~\ref{thm:maze}, if such a family $\maze$ is large enough, it contains an $a$-bowtie
(which is the conclusion of Theorem~\ref{thm:bowtie-cuts}) or a
flower of size $b$.
To conclude the proof of Theorem~\ref{thm:bowtie-cuts}, it remains to show the following.
\begin{lemma}\label{lem:flower}
Let $G$ be a directed graph with two distinguished vertices $s$ and $t$ and let $k$
be an integer. Let $\maze$ be a family of pairwise disjoint minimal $st$-cuts of size $k$
with the aforementioned structure of a $k$-maze.
Then every flower in $\maze$ has cardinality at most $(k+1)4^{k+1}$.
\end{lemma}

Lemma~\ref{lem:flower} is an easy corollary of the so-called \emph{anti-isolation lemma}
(see e.g.~\cite{DBLP:conf/soda/PilipczukW16} for a proof).
\begin{lemma}[Anti-isolation lemma]\label{anti-isolation}
Let $k$ be an integer, $G$ be a directed graph with a distinguished vertex $s$, and let
$T \subseteq V(G)$.
Assume that for every $t \in T$ there exists a set $C_t \subseteq E(G)$ of size at most $k$ such that
$t$ is the only vertex of $T$ reachable from $s$ in $G-C_t$. Then $|T| \leq (k+1)4^{k+1}$.
\end{lemma}

\begin{proof}[Proof of Lemma~\ref{lem:flower}.]
  Let $\mathcal{F} \subseteq \maze$ be a flower with a value $\zeta \in \{\bot, \top\}$ and an element $e(u) \in u$ for every $u \in \mathcal{F}$
  such that $f_{u,v}(e(u)) = \zeta$ for every $v \in \mathcal{F} \setminus \{u\}$.
  Observe that if we revert all the edges of $G$ and switch the roles of $s$, and $t$, $\maze$ remains a family of minimal $st$-cuts of size $k$,
  but the values of $\bot$ and $\top$ in functions $f_{u,v}$ swap. In particular, $\mathcal{F}$ remains a flower with the same choice of $e(u)$
    for every $u \in \mathcal{F}$, but the value of $\zeta$ changes.
  Hence, without loss of generality we can assume that $\zeta = \top$.

  Fix $u \in \mathcal{F}$.
  Let $t(u)$ be the tail of $e(u)$ for every $u \in \mathcal{F}$. Since $u$ is a minimal $st$-cut, $t(u)$ is reachable from $s$ in $G-u$.
  Since $f_{u,v}(e(u)) = \zeta = \top$ for every $v \in \mathcal{F} \setminus \{u\}$, $t(u)$ is not reachable from $s$ in $G-v$.
  In particular, the tails $t(u)$ are pairwise distinct for different $u \in \mathcal{F}$.
  Lemma~\ref{anti-isolation} implies that the set $\{t(u)~|~u \in \mathcal{F}\}$ is of cardinality at most $(k+1)4^{k+1}$, which
  implies the same bound on $|\mathcal{F}|$.
\end{proof}

\section{NP-hardness of \textsc{Multi-budgeted cut}}\label{sec:np}

Although it is well-known that the minimum cut problem is polynomial-time solvable, we prove that the \textsc{Multi-budgeted cut} problem is NP-hard for $\ell\geq 2$.
\begin{lemma}\label{NP_hard}
\textsc{Multi-budgeted cut} problem is NP-hard for every $\ell \geq 2$.
\end{lemma}
\begin{proof}
We prove this lemma by making a reduction from constrained minimum vertex cover problem on bipartite graphs (\textsc{Min-CBVC}),
   which is proved to be NP-hard by Chen and Kanj \cite{DBLP:journals/jcss/ChenK03}.
In the \textsc{Min-CBVC} problem the input consists of a bipartite graph $G=(U \uplus L, E)$ and integers $k_U,k_L$; the goal is to find a vertex cover $X \subseteq U \cup L$
such that $|X \cap U| \leq k_U$ and $|X \cap L| \leq k_L$.

We reduce to a \textsc{Multi-budgeted cut} instance with $\ell = 2$. For larger values of $\ell$, it is straightforward to pad the instance as follows: for every $3 \leq i \leq \ell$,
create a new edge $e_i$ with tail in $s$ and head in $t$ and set $E_i = \{e_i\}$, $k_i = 1$. 

Given an instance $(G,k_{U},k_{L})$ of \textsc{Min-CBVC}, where $G=(U\cup L,E)$ is a bipartite graph, we construct an instance $(G',X,Y,2,(E_i,k_i)_{i=1}^2)$ of \textsc{Multi-budgeted cut} as follows.
We take $V(G') = V(G) \cup \{s,t\}$, where $s$ and $t$ are two new vertices, and set $X=\{s\}$ and $Y=\{t\}$.
Then for each vertex $u\in U$, we add an arc $(s,u)$ with color $1$ to $G'$ and for each vertex $v\in L$, we add an arc $(v,t)$ with color $2$ to $G'$.
For each undirected edge $(u,v)\in E(G)$ such that $u\in U$ and $v\in L$, we add an arc $(u,v)$ with no color. Let $E_{1}$ be the set of arcs of color $1$ in $G'$ and $E_{2}$ be the set of arcs of color $2$ in $G'$. Let $Z=E_{1}\cup E_{2}$ be the deletable arcs. Let the budgets of the \textsc{Multi-budgeted cut} instance be $k_{1}=k_{U},k_{2}=k_{L}$. This completes the construction.

Now we show that $(G,k_{U},k_{L})$ is a yes instance if and only if $(G',X,Y,2,(E_i,k_i)_{i=1}^2)$ is a yes instance. Suppose $(G,k_{U},k_{L})$ is a yes instance. Then there exists a vertex cover $U'\cup L'$ of $G$ such that $U'\subseteq U$, $L'\subseteq L$, $|U'|\leq k_{U}$ and $|L'|\leq k_{L}$. Let $C_{1}=\{(s,u)|u\in U'\}$ and $C_{2}=\{(v,t)|v\in L'\}$. We claim that $C_{1}\cup C_{2}$ is a solution for $(G',X,Y,2,(E_i,k_i)_{i=1}^2)$. Obviously $|C_{1}|\leq k_{1}$, $|C_{2}|\leq k_{2}$ and $C$ is $Z$-respecting. For contradiction, suppose that there is a directed path $su'v't$ in $G'\setminus (C_{1}\cup C_{2})$. It follows that $u'\notin U'$ and $v'\notin L'$. Thus there is an edge $u'v'$ which is not covered by $U'\cup L'$ in $G$, contradicting that $U'\cup L'$ is a vertex cover of $G$.

Suppose that $(G',X,Y,2,(E_i,k_i)_{i=1}^2)$ is a yes instance. Then there is a $Z$-respecting budget-respecting $st$-cut $C=C_{1}\cup C_{2}$ such that $C_{1}$ is a set of arcs of color $1$ of size at most $k_{1}$ and $C_{2}$ is a set of arcs of color $2$ of size at most $k_{2}$.
Obviously any arc between $U$ and $V$ in $G'$ is not in the solution because they are not deletable. Let $U'=\{u|(s,u)\in C_{1}\}$ and $L'=\{v|(v,t)\in C_{2}\}$. We get that $U'\subseteq U$, $L'\subseteq L$, $|U'|\leq k_{1}=k_{U}$ and $|L'|\leq k_{2}=k_{L}$. We claim that $U'\cup L'$ is a solution for $(G,k_{U},k_{L})$. For contradiction, suppose that there is an edge $u'v'$ not covered by $U'\cup L'$. It follows that $su'v't$ is a directed path in $G'\setminus C$, contradicting that $C=C_{1}\cup C_{2}$ is a solution for $(G',X,Y,2,(E_i,k_i)_{i=1}^2)$. This completes the proof.
\end{proof}

\section{Conclusion}

We would like to conclude with a discussion on future research directions. First, our upper bound
  of $2^{\Oh(k^2 \log k)}$ 
  on the number of multi-budgeted important separators (Theorem~\ref{thm:imp-enum}) is far from the $4^k$ bound for the classic important separators. As pointed out by an anonymous reviewer 
  at IPEC 2018, there is an easy lower bound of $k!$ for the number of multi-budgeted
  important separators: Let $\ell = k$, $k_i=1$ for every $i \in [\ell]$, and let $G$
  consist of $k$ paths from $s$ to $t$, each path consisting of $\ell$ edges of different colors.
  Then there are exactly 
  $k!$ distinct multi-budgeted important separators, as we can freely choose a different color
  $i \in [\ell]$ to cut on each path.
  We are not aware of any better lower bound, leaving a significant gap between the lower
  and upper bounds.

Second, our existential statement of Theorems~\ref{thm:bnd-weights} and~\ref{thm:bnd-chain}
can be treated as a weak support of tractability of \textsc{Chain $\ell$-SAT} and 
  \textsc{Weighted $st$-cut}. Are they really FPT when parameterized by the cardinality of the cut?


\bibliography{ColorCuts}
\end{document}